\documentclass[a4paper,12pt]{article}

\usepackage{amsmath,amsthm,amssymb,amsfonts}

\makeatletter

\@addtoreset{equation}{section}
\makeatother

\newtheorem{theorem}{Theorem}[section]
\newtheorem{lemma}[theorem]{Lemma}
\newtheorem{proposition}[theorem]{Proposition}

\def\ep{\varepsilon}

\def\R{\mathbb R}
\def\S{\mathbb S}
\def\N{\mathbb N}
\def\pa{\partial}
\def\b{\backslash}

\begin{document}

\title{Inverse scattering at fixed energy for the multidimensional Newton equation in short range radial potentials}
\author{Alexandre Jollivet}

\maketitle

\begin{abstract}
We consider the inverse scattering problem at fixed and sufficiently large energy for the nonrelativistic and relativistic Newton equation in $\R^n$, $n \ge 2$, with a smooth and short range electromagnetic field $(V,B)$.
Using results of [Firsov, 1953] or [Keller-Kay-Shmoys, 1956] we obtain a uniqueness result when $B$ is assumed to be zero  in a neighborhood of infinity and $V$ is assumed to be spherically symmetric in a neighborhood of infinity. 
\end{abstract}

\section{Introduction}
Consider the following second order differential equation that is the multidimensional nonrelativistic Newton equation with electromagnetic field
\begin{eqnarray}
\ddot x(t)=F(x(t),\dot x(t)):=-\nabla V(x(t))+B(x(t))\dot x(t),\label{1.1}
\end{eqnarray}
where $x(t)\in \R^n,$ $\dot x(t)={{\rm d}x\over {\rm d}t}(t)$. In this equation we assume that $V\in C^2(\R^n,\R)$ and for any $x\in \R^n,$ $B(x)$ is a $n\times n$
antisymmetric matrix with elements $B_{i,k}(x),$  $B_{i,k}\in C^1(\R^n,\R)$, which satisfy
\begin{equation}
{\pa B_{i,k}\over \pa x_l}(x)+{\pa B_{l,i}\over \pa x_k}(x)+{\pa B_{k,l}\over \pa x_i}(x)=0,\label{1.4B}
\end{equation}
for $x=(x_1,\ldots,x_n)\in \R^n$ and for $l,i,k=1\ldots n$.

For $n=3$, the equation \eqref{1.1} is the equation of motion in $\R^n$ of a nonrelativistic particle of mass $m=1$ and charge $e=1$ in an
external and static electromagnetic field described by $(V,B)$ (see, for example, \cite[Section 17]{[LL2]}).
For the electromagnetic field the function $V$ is an electric potential and $B$ is the
magnetic field. Then $x$ denotes the position of the particle,
$\dot x$ denotes its velocity, $\ddot x$ denotes its acceleration and $t$ denotes the time.

For the equation \eqref{1.1} the energy
\begin{equation}
E={1\over 2}|\dot x(t)|^2+V(x(t))\label{1.2}
\end{equation}
is an integral of motion.

We assume that the electromagnetic coefficients $V$ and $B$ are short range.  More precisely we assume that $(V,B)$ satisfies the following conditions
\begin{eqnarray}
|\pa_x^{j_1}V(x)|&\le& \beta_{|j_1|}(1+|x|)^{-\alpha-|j_1|},\ x\in \R^n,\label{1.3}\\
|\pa_x^{j_2}B_{i,k}(x)|&\le& \beta_{|j_2|+1}(1+|x|)^{-\alpha-1-|j_2|},\ x\in \R^n,\label{1.4}
\end{eqnarray}
for $|j_1|\le 2,|j_2|\le 1,$ $i,k=1\ldots n$ and some $\alpha>1$ (here $j_l$ is the multiindex $j_l=(j_{l,1},\ldots,j_{l,n})
\in (\N\cup \{0\})^n,$ $|j_l|=\sum_{k=1}^nj_{l,k}$ and $\beta_{|j_l|}$ are positive real
constants).
We denote by $\|.\|$ the norm on the short range electromagnetic fields defined by
\begin{eqnarray}
\|(V,B)\|&=&\sup_{x\in \R^n, \ j_1\in \N^n\atop |j_1|\le 2}\Big((1+|x|)^{\alpha+|j_1|}|\pa_x^{j_1}V(x)|\Big)\label{1.4c} \\
&&+\sup_{x\in \R^n,\ j_2\in \N^n\atop |j_2|\le 1,\ i,k=1\ldots n}\Big((1+|x|)^{\alpha+1+|j_2|}|\pa_x^{j_2}B_{i,k}(x)|\Big).\nonumber
\end{eqnarray}

Under conditions \eqref{1.3}--\eqref{1.4}, we have the following properties (see, for example,  \cite{[S]} and
\cite{[LT]} where classical scattering of particles in a short-range electric field and in a long-range magnetic field are studied respectively): for any
$(v_-,x_-)\in \R^n\times\R^n,\ v_-\neq 0,$
the equation \eqref{1.1}  has a unique solution $x\in C^2(\R,\R^n)$ such that
\begin{equation}
{x(t)=tv_-+x_-+y_-(t),}\label{1.6}
\end{equation}
where $|\dot y_-(t)|+|y_-(t)|\to 0,\ {\rm as}\ t\to -\infty;$  in addition for almost any
$(v_-,x_-)\in \R^n\times \R^n,\ v_-\neq 0,$ the unique solution $x(t)$ of equation \eqref{1.1} that satisfies \eqref{1.6} also satisfies the following asymptotics
\begin{equation}
{x(t)=tv_++x_++y_+(t),}\label{1.7}
\end{equation}
 where $v_+\neq 0,\ |\dot y_+(t)|+|y_+(t)|\to 0,{\rm\ as\ }t \to +\infty$.
At fixed energy $E>0$, we denote by $\S^1_E$ the set $\{v_-\in \R^n\ |\ |v_-|^2=2E\}$ and  we denote by ${\cal D}(S_E)$ the set of $(v_-,x_-)\in \S^1_E\times \R^n$ for which the unique solution $x(t)$ of equation \eqref{1.1} that satisfies \eqref{1.6} also satisfies \eqref{1.7}.
We have that ${\cal D}(S_E)$  is an open set of $\S^1_E \times \R^n$ and ${\rm Mes}((\S^1_E \times \R^n) \b {\cal D}(S_E))=0$ for the Lebesgue
measure on $\S^1_E \times \R^n$.
The map
$S_E: {\cal D}(S_E) \to \S^1_E\times\R^n $
given by the formula
\begin{equation}
S_E(v_-,x_-)=(v_+,x_+),\label{1.8}
\end{equation}
is called the scattering map at fixed energy $E>0$ for the equation \eqref{1.1}. Note that
if $V(x)\equiv 0$ and $B(x)\equiv 0$, then $v_+=v_-,\ x_+=x_-,\ (v_-,x_-)\in \R^n \times \R^n,\ v_-\neq 0$.

In this paper we consider the following inverse scattering problem  at fixed energy
\begin{equation}
\textrm{Given }S_E \textrm{ at fixed energy }E>0,\ \textrm{find }(V,B).\label{1.9}
\end{equation}

Note that using the conservation of energy we obtain that  if $E < \sup_{\R^n}V$ then $S_E$ does not
determine uniquely $V$.

We mention results on Problem \eqref{1.9}.
When $B\equiv0$ and $V$ is assumed to be spherically symmetric and monotonuous decreasing in $|x|$ ($V$ is not assumed to be short range), uniqueness results for Problem \eqref{1.9} were obtained  in \cite{[F],[KKS]}.
The scattering map  $S_E$ also uniquely determines $(V,B)$ at fixed and sufficiently large energy when $(V,B)$ is assumed to be compactly supported inside a fixed domain of $\R^n$ (see \cite{[No]} for $B\equiv0$ and see \cite{[Jo3]}). This latter result relies on a uniqueness result for an inverse boundary kinematic problem for equation \eqref{1.1} (see \cite{[GN],[No]} when $B\equiv0$, and see \cite{[DPSU], [Jo3]}) and connection between this boundary value problem and the inverse scattering problem on $\R^n$ (see \cite{[No]} for $B\equiv0$, and see \cite{[Jo3]}).

To our knowledge it is still unknown whether the scattering map at fixed and sufficiently large energy uniquely determine the electromagnetic field under the regularity and short range conditions \eqref{1.3} and \eqref{1.4} (see \cite[Conjecture B]{[No]} for $B\equiv 0$).

In this paper we propose a generalization of results in \cite{[F],[KKS]} for the short range case where no decreasing monotonicity is assumed. More precisely we have the following uniqueness result.

\begin{theorem}
\label{thm}
Let $(\lambda,R)\in (0,+\infty)^2$ and let $(V,B)$ be an electromagnetic field that satisfies the assumptions \eqref{1.4B}, \eqref{1.3} and \eqref{1.4} and $\|(V,B)\|\le \lambda$.
Assume that $B\equiv 0$ outside $B(0,R)$ and that $V$ is spherical symmetric outside $B(0,R)$.  Then there exists a positive constant $E(\lambda,R)$ (which does not depend on $(V,B)$) so that the scattering map at fixed energy $E>E(\lambda, R)$ uniquely determines $(V,B)$ on $\R^n$.
\end{theorem}

The proof of Theorem \ref{thm} is obtained by recovering first the electric potential in a neighborhood of infinity using Firsov or Keller-Kay-Shmoys' result \cite{[F], [KKS]}  and then by recovering the electromagnetic field on $\R^n$ using the following proposition  which generalizes \cite[Theorem 7.2]{[Jo3]}.

\begin{proposition}
\label{prop}
Let $(\lambda,R)\in (0,+\infty)^2$ and let $(V,B)$ be an electromagnetic field that satisfies the assumptions \eqref{1.4B}, \eqref{1.3} and \eqref{1.4} and $\|(V,B)\|\le \lambda$.
Assume that $(V,B)$ is known outside $B(0,R)$. Then there exists a positive constant $E(\lambda,R)$ so that the scattering map at fixed energy $E>E(\lambda,R)$ uniquely determines $(V,B)$ on $\R^n$.
\end{proposition}

Concerning the inverse scattering problem for the classical multidimensional nonrelativistic Newton equation at high energies and the  inverse scattering problem for a particle in electromagnetic field (with $B\not\equiv 0$ or $B\equiv0$) in quantum
mechanics, we refer the reader to
\cite{[GN],[No],[Jo3],[Jo4]} and references therein.

Concerning the inverse problem for \eqref{1.1} in the one-dimensional case, we can mention the works \cite{[Ab],[K],[AFC]}.

The structure of the paper is as follows. In section \ref{sec_prop} we prove Proposition \ref{prop}. In section \ref{sec_thm} we prove Theorem \ref{thm}.
In section \ref{sec_rel} we provide  similar results for the relativistic multidimensional Newton equation with electromagnetic field.

\section{Proof of Proposition \ref{prop}}
\label{sec_prop}
\subsection{Nontrapped solutions of equation \eqref{1.1}}
We will use the  standard Lemma \ref{l1.4.1} on nontrapped solutions of equation \eqref{1.1}. For sake of consistency its proof is given in Appendix.

\begin{lemma}
\label{l1.4.1}
Let  $E>0$ and let $R_E$ and $C_E$ be defined by
\begin{eqnarray}
&&C_E:={2E\over (n\beta_1+2\beta_0)(1+\sqrt{2(E+\beta_0)})},\label{112}\\
&&\sup_{|x|\ge R_E}(1+|x|)^{-\alpha}\le {C_E\over 2}.\label{111}
\end{eqnarray}
If $x(t)$ is a solution of equation (1.1) of energy $E$ such that $|x(0)|<R_E$ and if there exists a time $T>0$ such that $x(T)=R_E$ then
\begin{equation}
|x(t)|^2\ge R_E^2+ E |t-T|^2\textrm{ for }t\in (T,+\infty),\label{114}
\end{equation}
and there exists a unique $(x_+,v_+)\in \R^n\times\S^1_E$ so that
$$
x(t)=x_++tv_++y_+(t),\ t\in \R,
$$
where $|y_+(t)|+|\dot y_+(t)|\to 0$ as $t\to +\infty.$
\end{lemma}

Note that $C_E\to +\infty$ as $E\to +\infty$ while $\sup_{|x|\ge R}(1+|x|)^{-\alpha}$ is a decreasing function of $R$ that goes to $0$ as $R\to +\infty$.
Note that Lemma \ref{l1.4.1} is stated for positive times $t$ but a similar result hold for negative times $t$.

\subsection{The inverse kinematic problem for equation \eqref{1.1}}
\label{invkin}
We first formulate the inverse kinematic problem for equation \eqref{1.1} inside a ball of center $0$ and radius $R>0$ denoted by $B(0,R)$. 
For $(m,l)\in (\N\b\{0\})^2$ and for a function $f$ from $B(0,R)$ to $\R^m$ of class $C^l$ we define the $C^l$ norm of $f$ by
$$
\|f\|_{C^l,R}=\sup_{x\in B(0,R),\ \alpha\in \N^n\atop |\alpha|\le l}|\pa_x^{\alpha}f(x)|.
$$
We denote by $\pa B(0,R)$ the boundary of the ball $B(0,R)$.

Then we recall that there exists a constant $E(R,\|V\|_{C^2,R}, \|B\|_{C^1,R})$
so that at fixed energy $E>E(R,\|V\|_{C^2,R}, \|B\|_{C^1,R})$ the solutions $x$ of equation \eqref{1.1} in $B(0,R)$ at energy
$E$ have the following properties (see for example \cite{[Jo3]}):
\begin{equation}
\label{1.4a}\begin{array}{l}
\textrm{for each solution }x(t) \textrm{ there are }t_1,t_2\in \R, t_1<t_2, \textrm{ such that }\\
x\in C^3([t_1,t_2],\R^n),\ (x(t_1), x(t_2))\in \pa B(0,R)^2,\ x(t)\in B(0,R)
\textrm{ for }t\in ]t_1,t_2[,\\
x(s_1)\not=x(s_2) {\rm \ for\ } s_1,s_2\in [t_1,t_2], s_1\not= s_2;
\end{array}
\end{equation}
and
\begin{equation}
\label{1.5a}\begin{array}{l}
\textrm{for any two distinct points }q_0,q\in \pa B(0,R), \textrm{ there is one and only one solution}\\
x(t)=x(t,E,q_0,q)\textrm{ such that }x(0)=q_0, x(s)=q \textrm{ for some }s>0.
\end{array}
\end{equation}
This is closely related to the property that at fixed and sufficiently large energy $E$, the compact set $\overline{B(0,R)}$ endowed with the riemannian metric $\sqrt{E-V(x)}|dx|$ and the magnetic field defined by $B$ is simple (see \cite{[DPSU]}).

For $(q_0,q)$ two distinct points of $\pa B(0,R)$ we denote by $s(E,q_0,q)$ the time at which $x(t,E,q_0,q)$ reaches $q$ from $q_0$ and we denote by
 $k_0(E,q_0,q)$ the velocity vector $\dot x(0,E,q_0,q)$
and by  $k(E,q_0,q)$ the velocity vector $\dot x(s(E,q_0,q),$ $E,q_0,q).$
The inverse kinematic problem is then
\begin{eqnarray*}
\textrm{ Given }k(E,q_0,q),\ k_0(E,q_0,q) \textrm{ for all } q_0,q\in\pa B(0,R),&&\\
q_0\not=q,\textrm{ at fixed
 sufficiently large energy }E,
\textrm{ find } (V,B) \textrm{ in }B(0,R).&&
\end{eqnarray*}

The data $k_0(E,q_0,q),$ $k(E,q_0,q),$ $q_0,q\in\pa B(0,R), q_0\not=q,$ are the boundary value data of the inverse kinematic problem, and we recall the following result.

\begin{lemma}[see, for example, Theorem 7.1 in \cite{[Jo3]}]
\label{inv}
At fixed $E>E(\|V\|_{C^2,R},$ $\|B\|_{C^1,R},R)$, the boundary data $k_0(E,q_0,q)$,
$(q_0,q)\in \pa B(0,R)\times \pa B(0,R),$ $q_0\not=q$,
uniquely determine $(V,B)$ in $B(0,R)$.
\end{lemma}

\subsection{Relation between boundary data of the inverse kinematic problem and the scattering map $S_E$}
\label{connect}
We will prove that at fixed and sufficiently large energy $E$ the scattering map $S_E$ determines the boundary data $k_0(E,q_0,q),$ $k(E,q_0,q),$ $q_0,q\in\pa B(0,R), q_0\not=q$. This will prove that $S_E$ uniquely determines $(V,B)$ in $B(0,R)$, which will prove Proposition \ref{prop}.

Let $R>0$ and $\lambda>0$ be such that $(V,B)$ are known outside $B(0,R)$ and $\|(V,B)\|<\lambda$. 
Note that $\max(\|V\|_{C^2,R}, \|B\|_{C^1,R})\le \|(V,B)\|<\lambda$. Thus there exists a constant $E_0(\lambda,R)$ such that at fixed energy $E>E_0(\lambda,R)$ solutions $x(t)$ of equation \eqref{1.1} in $B(0,R)$ at energy $E$ have properties \eqref{1.4a} and  \eqref{1.5a} and such that at fixed $E>E_0(\lambda, R)$ the boundary data $k_0(E,q_0,q)$,
$(q_0,q)\in \pa B(0,R)\times \pa B(0,R),$ $q_0\not=q$,
uniquely determine $(V,B)$ in $B(0,R)$.
Then using that the constant $C_E\to +\infty$ as $E\to +\infty$ in Lemma \ref{l1.4.1} we obtain that there exists $E_1(\lambda,R)$ such that for $E>E_1(\lambda,R)$ we have $\sup_{|x|\ge R}(1+|x|)^{-\alpha}\le {C_E\over 2}$ so that $R_E$ can be replaced by $R$ in Lemma \ref{l1.4.1}.

Set $E(\lambda, R)=\max(E_0(\lambda,R),E_1(\lambda,R))$ and fix $E>E(\lambda, R)$. 
Let $(x_-,v_-)\in {\cal D}(S_E)$ and $(v_+,x_+)=S_E(v_-,x_-)$. We denote by $x(.,v_-,x_-)$ the solution of equation \eqref{1.1} that satisfies \eqref{1.6} (and \eqref{1.7}).
Set
\begin{eqnarray*}
t_-(x_-,v_-)&=&\sup\{t\in \R\ |\ |x(s,x_-,v_-)|\ge R,\ s\in (-\infty,t)\},\\
t_+(x_-,v_-)&=&\inf\{t\in \R\ |\ |x(s,x_-,v_-)|\ge R,\ s\in  (t,+\infty)\}.
\end{eqnarray*}
Since $(V,B)$ is known outside $B(0,R)$ we can solve equation \eqref{1.1} with initial conditions \eqref{1.6} and \eqref{1.7} and we obtain that $x(.,x_-,v_-)$ is known on $(-\infty,t_-(x_-,v_-)]\cup [t_+(x_-,v_-),\infty)$.

If $x(s,x_-,v_-)\not\in B(0,R)$ for any $s\in \R$, then $t_\pm(x_-,v_-)=\mp\infty$.
If there exists $s\in \R$ such that $x(s,x_-,v_-)\in B(0,R)$ then set 
\begin{equation}
q_0=x(t_-(x_-,v_-)),\ q=x(t_+(x_-,v_-)).\label{300}
\end{equation}
Using Lemma \ref{l1.4.1} and $E>E_1(\lambda, R)$ we obtain
that $|x(s,x_-,v_-)|<R$ for $s\in(t_-(x_-,v_-),t_+(x_-,v_-))$ and $q_0\not=q$. (Note that if $x(t)$ satisfies equation \eqref{1.1} then $x(t+t_0)$ also satisfies \eqref{1.1} for any $t_0\in \R$.)
Therefore we have $x(s,x_-,v_-)=x(s-t_-(x_-,v_-),E,q_0,q)$ for $s\in (t_-(x_-,v_-),t_+(x_-,v_-))$ where $x(t,E,q_0,q)$ is the solution of \eqref{1.1} given by \eqref{1.5a}, and we have 
\begin{equation}
k_0(E,q_0,q)=\dot x(t_-(x_-,v_-)),\ k(E,q_0,q)=\dot x(t_+(x_-,v_-)).\label{301}
\end{equation}
We proved that the scattering map $S_E$ uniquely determines the data $k_0(E,q_0,q),$ $k(E,q_0,q),$ $(q_0,q)\in\pa B(0,R)^2, q_0\not=q$, when $(q_0,q)=(x(t_-(x_-,v_-)),x(t_+(x_-,v_-)))$ for 
$(x_-,v_-)\in {\cal D}(S_E)$. 
And using again Lemma \ref{l1.4.1} we know that for any $(q_0,q)\in \pa B(0,R)^2$, $q_0\not=q$, the solution $x(t,E,q_0,q)$ given by \eqref{1.5a} satisfies 
\eqref{1.6} and \eqref{1.7} for some $(x_\pm,v_\pm)\in \R^n\times \S^1_E$ and that $|x(t,E,q_0,q)|>R$ for $t<0$ and $t>s(E,q_0,q)$.
Thus $S_E$ uniquely determines the data $k_0(E,q_0,q),$ $k(E,q_0,q),$ $(q_0,q)\in\pa B(0,R)^2, q_0\not=q$.\hfill $\Box$

\section{Proof of Theorem \ref{thm}}
\label{sec_thm}
In this section we assume that the electromagnetic field $(V,B)$ in equation \eqref{1.1} satisfies \eqref{1.3} and \eqref{1.4} and is so that $B\equiv 0$ and $V$ is spherically symmetric outside $B(0,R)$ for some $R>0$. Let $W\in C^2([R,+\infty),\R)$ be defined by $V(x)=W(|x|)$ for $x\not\in B(0,R)$.
From \eqref{1.3} it follows that
\begin{equation}
\sup_{r>R}(1+r)^\alpha|W(r)|\le \beta_0\textrm{ and } \sup_{r>R}(1+r)^{\alpha+1}|W'(r)|\le \beta_1,\label{1.3W}
\end{equation}
where $W'$ denotes the derivative of $W$.

\subsection{The function $r_{\min,.}$}
Set
\begin{equation}
\beta: =(2E+2\beta_0)^{1\over 2}\max\big(R,\big({\beta_1+2\beta_0\over 2E}\big)^{1\over \alpha}\big).\label{4.8}
\end{equation}
Then for $q\ge\beta$ consider the real number $r_{\min,q}$ defined by 
\begin{equation}
r_{\min,q}=\sup\{r\in  (R,+\infty)\ | \ W(r)+{q^2\over 2r^2}=E\}.\label{4.2}
\end{equation}
The function $r_{\min,.}$ has the following properties.

\begin{lemma}
\label{lemrmin}
The function $r_{\min,.}$ is a $C^2$ strictly increasing function from $[\beta,+\infty)$ to $(R,+\infty)$ and we have  $r_{\min,q}^3W'(r_{\min,q})<q^2$ for $q\ge\beta$ and
\begin{equation}
W(r_{\min,q})+{q^2\over 2r_{\min,q}^2}=E,\ {d r_{\min,q}\over dq}={qr_{\min,q}\over q^2-r_{\min,q}^3W'(r_{\min,q})}>0.\label{4.7b}
\end{equation}
In addition the following estimates and asymptotics at $+\infty$ hold
\begin{equation}
{q\over \sqrt{2E+2\beta_0}} \le r_{\min,q} \le {q\over \Big( 2E-2\beta_0q^{-\alpha}(2\beta_0+2E)^{\alpha\over2}\Big)^{1\over 2}},\ \textrm{ for }\ q\in [\beta,+\infty),\label{4.6a}
\end{equation}
\begin{equation}
r_{\min,q}={q\over \sqrt{2E}}+O(q^{1-\alpha}),\textrm{ as }q\to +\infty.\label{4.6c}
\end{equation}
\end{lemma}

Lemma \ref{lemrmin} and works \cite{[F],[KKS]} allow the reconstruction of the force $F$ in a neighborhood of infinity. In sections \ref{scatangle} and \ref{reconst} we develop the reconstruction procedure given in \cite{[F],[KKS]}.

\subsection{The scattering angle}
\label{scatangle}
Let ${\cal P}$ be a plane of $\R^n$ containing $0$ and let $(e_1,e_2)$ be an orthonormal basis of ${\cal P}$.
For $(v_1,v_2)\in \R^2$ and for $v=v_1e_1+v_2e_2$ we define $v^\bot\in{\cal P}$ by $v^\bot=-v_2e_1+v_1e_2$.

Let $q\ge\beta$.
Then set $x_-:=(2E)^{-{1\over 2}}qe_1$ and $v_-=\sqrt{2E}e_2$. We have $x_-\cdot v_-=0$ and $x_-\cdot v_-^\bot=-q$, and for such couple $(x_-,v_-)$ we consider $x_q(t)$ the solution of \eqref{1.1} with energy $E$ and with initial conditions \eqref{1.3} at $t\to -\infty$. Let $t_-:=\sup\{t\in \R\ |\ |x_q(s)|\ge R \textrm{ for }s\in (-\infty,t)\}$. We will prove that $t_-=+\infty$. Since the force $F$ in \eqref{1.1} is radial outside $B(0,R)$ we obtain that $x_q(t)\in {\cal P}$ for $t\in (-\infty,t_-)$.

We introduce polar coordinates in ${\cal P}$. We write $x_q(t)=r_q(t)(\cos(\theta_q(t))e_1(t)+\sin(\theta_q(t))e_2)$ for $t\in (-\infty,t_-)$ where the functions $r_q$, $\theta_q$, satisfy the following ordinary differential equations
\begin{eqnarray}
\ddot r_q(t)&=&-W'(r_q(t))+{q'^2\over r_q(t)^3},\label{2.3e}\\
r_q(t)^2\dot \theta_q(t)&=&q',\textrm{ for some }q'\in \R.\label{2.3f}
\end{eqnarray}
Asymptotic analysis of $x_q(t)$ at $t=-\infty$ using the initial conditions \eqref{1.3} and $\dot y_-(t)=o(t^{-1})$ as $t\to -\infty$ (see for example \cite[Theorem 3.1]{[No]} for this latter property) shows that $q'=q$. We refer the reader to the Appendix for details.

The energy $E$ defined by \eqref{1.2} is then written as follows
\begin{equation}
2E=\dot r_q(t)^2+{q^2\over r_q(t)^2}+2W(r_q(t)).\label{2.5}
\end{equation}

Let $t_q=\inf\{t\in (-\infty,t_-)\ |\ r_q(t)=r_{\min,q}\}$. Then using \eqref{4.2} and \eqref{2.5} we have $\dot r_q(t_q)=0$. Since $r_q$ satisfies the second order differential equation \eqref{2.3e} we obtain that $t_-=+\infty$, $r_q(t_q+t)=r_q(t_q-t)$ for $t\in\R$, and $\pm\dot r_q(t)>0$ for
$\pm t>t_q$. We thus define
\begin{equation}
g(q)=\int_{-\infty}^{+\infty}{dt\over r_q(t)^2}=2\int_{t_q}^{+\infty}{dt\over r_q(t)^2}.\label{4.0}
\end{equation}
The integral \eqref{4.0} is absolutely convergent and from \eqref{2.3f} it follows that $qg(q)=\int_\R\dot \theta_q(s) ds$ is the scattering angle of $x_q(t)$, $t\in \R$.

Note also that from \eqref{2.3f} we have
\begin{equation}
S_{1,E}(\sqrt{2E}e_2,(2E)^{-{1\over 2}}qe_1)=\sqrt{2E}\Big(\cos\big(qg(q)+{\pi\over 2}\big),\sin\big(qg(q)+{\pi\over 2}\big)\Big),\label{4.104a}
\end{equation}
for $q\in [\beta,+\infty)$ and where $S_{1,E}$ is the first component of the scattering map $S_E$.
Note that $qg(q)\to 0$ as $q\to +\infty$ and  that $g$ is continuous on $[\beta,+\infty)$  (these properties can be proven by using \cite[Theorem 3.1]{[No]} and continuity of the flow of \eqref{1.1}).
Hence using \eqref{4.104a} we obtain that $S_{1,E}$ uniquely determines the function $g$.

\subsection{Reconstruction formulas}
\label{reconst}
Let $\chi$ be the strictly increasing function from $[0,\beta^{-2})$ to $[0,r_{\min,\beta}^{-1})$, continuous on $[0,\beta^{-2})$ and $C^2$ on $(0,\beta^{-2})$, defined by
\begin{equation}
\chi(0)=0,\ \textrm{ and }\ \chi(\sigma)=r_{\min,\sigma^{-{1\over 2}}}^{-1},\ \textrm{ for }\ \sigma\in (0,\beta^{-2}).\label{4.10}
\end{equation}
Let $\phi:(0,\chi(\beta^{-2}))\to (0,\beta^{-2})$ denote the inverse function of $\chi$.
From \eqref{4.7b} and \eqref{4.10} it follows that
\begin{eqnarray}
2(E-V(\chi(u)^{-1}))&=&{\chi(u)^2\over u},\ \textrm{for}\ u\in (0,\beta^{-2}),\label{4.203b}\\
2(E-V(s^{-1}))\phi(s)&=&s^2,\ \textrm{for}\ s\in (0,\chi(\beta^{-2})).\label{4.203}
\end{eqnarray}

Define the function $H$ from $(0,\beta^{-2})$ to $\R$ by
\begin{equation}
H(\sigma):=\int_0^\sigma {g(u^{-{1\over 2}})du\over 2\sqrt{u}\sqrt{\sigma-u}}\ \textrm{ for }\ \sigma\in (0,\beta^{-2}), \label{4.11a}
\end{equation}
Hence $H$ is known from the first component of the scattering map $S_E$.

The following formulas are valid (see Appendix for more details)
\begin{equation}
H(\sigma)
=\pi\int_0^{\chi(\sigma)}{ds\over \sqrt{2(E-V(s^{-1}))}},\ 
{1\over \pi\sqrt{\sigma}}{dH\over d\sigma}(\sigma)={d\over d\sigma}\ln(\chi(\sigma)),\label{4.11c}
\end{equation}
for $\sigma\in (0,\beta^{-2})$.

Then note that from \eqref{4.6c} it follows that
$\chi(\sigma)=\Big((2E)^{1\over 2}\sigma^{-{1\over 2}}+O(\sigma^{-1+\alpha\over 2})\Big)^{-1}=(2E)^{-{1\over 2}}\sigma^{1\over 2}+O(\sigma^{\alpha+1\over 2})$ as $\sigma\to 0^+$, and
$
\ln \Big({(2E)^{1 \over 2}\chi(\sigma)\sigma^{-{1\over 2}}}\Big)=\ln(1+O(\sigma^{\alpha\over 2}))\to 0$ as $\sigma\to 0^+$ 
(note that we just need the assumption $\alpha >0$).
Therefore we obtain the following reconstruction formulas
\begin{eqnarray}
\chi(\sigma)&=&(2E)^{-{1\over 2}}\sigma^{1\over 2}e^{\int_0^\sigma\big({1\over \pi\sqrt{s}}{dH\over ds}(s)-{1\over 2s}\big)ds}\ \textrm{for}\ \sigma\in (0,\beta^{-2}),
\label{4.12a}\\
W(s)&=&E-{1\over 2s^2\phi(s^{-1})}\ \textrm{for}\ s\in(r_{\min,\beta},+\infty).\label{4.12b}
\end{eqnarray}

Set
\begin{equation}
\beta'={\beta\over \Big( 2E-2\beta_0\beta^{-\alpha}(2\beta_0+2E)^{\alpha\over2}\Big)^{1\over 2}}.\label{4.12}
\end{equation}
Then note that from \eqref{4.6a} and \eqref{4.8}, it follows that
\begin{equation}
r_{\min,\beta}\le \beta'.\label{4.13}
\end{equation}
Therefore from \eqref{4.12a} and \eqref{4.12b} we obtain that $W$ is determined by the first component of the scattering map $S_E$ on $(\beta',+\infty)$.

The proof of Theorem \ref{thm} then relies on this latter statement and on Proposition \ref{prop}. \hfill $\Box$

\section{The relativistic multidimensional Newton equation}
\label{sec_rel}

\subsection{Uniqueness results}
Let $c>0$. 
Consider the relativistic multidimensional Newton equation in an electromagnetic field
\begin{eqnarray}
\dot p &=&F(x,\dot x):=-\nabla V(x)+{1\over c}B(x)\dot x,\label{6.1}\\
p&=&{\dot x \over \sqrt{1-{|\dot x|^2 \over c^2}}},\ \dot p={dp\over dt},\ \dot x={dx\over dt},\ x\in C^2(\R,\R^n),\nonumber
\end{eqnarray}
where $(V,B)$ satisfy \eqref{1.4B}, \eqref{1.3} and \eqref{1.4}.
The equation \eqref{6.1} is an equation for $x=x(t)$ and is the equation of motion in $\R^n$ of a relativistic particle of mass $m=1$
and charge $e=1$ in an external static electromagnetic field described by the scalar potential $V$ and the magnetic field $B$ (see \cite{[E]} and, for example, \cite[Section 17]{[LL2]}). In this equation
$x$ is the position of the particle, $p$ is its impulse, $F$ is the force acting on the particle, $t$ is the time and $c$ is the speed of light.

For the equation \eqref{6.1} the energy
\begin{equation}
E=c^2\sqrt{1+{|p(t)|^2 \over c^2}}+V(x(t))={c^2\over \sqrt{1-{|\dot x(t)|^2\over c^2}}}+V(x(t)),\label{6.2}
\end{equation}
is an integral of motion.
We denote by $B_c$ the euclidean open ball whose radius is c and whose centre is 0.

Under the conditions \eqref{1.2}, we have the following properties (see \cite{[Y]}): for any 
$(v_-,x_-)\in B_c\times\R^n,\ v_-\neq 0,$
the equation \eqref{6.1}  has a unique solution $x\in C^2(\R,\R^n)$ that satisfies \eqref{1.6} 
where $y_-$ in \eqref{1.6} satisfies $|\dot y_-(t)|+|\ y_-(t)|\to 0,$ as $t\to -\infty$;  in addition for almost any 
$(v_-,x_-)\in B_c\times \R^n,\ v_-\neq 0,$
the unique solution $x(t)$ of equation \eqref{1.1} that satisfies \eqref{1.6} also satisfies the asymptotics
\begin{equation}
{x(t)=tv_++x_++y_+(t),}\label{6.7}
\end{equation}
where $v_+\neq 0$, $|\dot y_+(t)|+|y_+(t)|\to 0$, as $t \to +\infty$.
At fixed energy $E>c$, we denote by $\S_{E,c}$ the set $\{v_-\in \R^n\ |\ |v_-|=c\sqrt{1-{c^4\over E^2}}\}$ and  we denote by ${\cal D}(S_E^{\rm rel})$ the set of $(v_-,x_-)\in \S_{E,c}\times \R^n$ for which the unique solution $x(t)$ of equation \eqref{6.1} that satisfies \eqref{1.6} also satisfies \eqref{1.7}.
We have that ${\cal D}(S_E^{\rm rel})$  is an open set of $\S_E \times \R^n$ and ${\rm Mes}((\S_{E,c} \times \R^n) \b {\cal D}(S_E^{\rm rel}))=0$ for the Lebesgue
measure on $\S_{E,c} \times \R^n$.
The map
$S_E^{\rm rel}: {\cal D}(S_E^{\rm rel}) \to \S_{E,c}\times\R^n $
given by $
S_E^{\rm rel}(v_-,x_-)=(v_+,x_+),$
is called the scattering map at fixed energy $E>c^2$ for the equation \eqref{6.1}. Note that
if $V(x)\equiv 0$ and $B(x)\equiv 0$, then $v_+=v_-,\ x_+=x_-,\ (v_-,x_-)\in B_c \times \R^n,\ v_-\neq 0$.

We consider the  inverse scattering problem  at fixed energy for  equation \eqref{6.1} that is similar to the inverse problem \eqref{1.9}
\begin{equation}
\textrm{Given }S_E^{\rm rel} \textrm{ at fixed energy }E>c^2,\ \textrm{find }(V,B).\label{6.9}
\end{equation}
Note that using the conservation of energy we obtain that  if $E < c^2+ \sup_{\R^n}V$ then $S_E$ does not
determine uniquely $V$.

For problem \eqref{6.9}  Theorem \ref{thm} and Proposition \ref{prop} still hold. 
In Sections \ref{prop_rel} and \ref{thm_rel} we sketch the proof of Theorem \ref{thm} and Proposition \ref{prop} for equation \eqref{6.1}.

For inverse scattering at high energies for the relativistic multidimensional Newton equation and inverse scattering in relativistic quantum mechanics see  \cite{[Jo2]} and references therein.

Concerning the inverse problem for \eqref{6.1} in the one-dimensional case, we can mention the work \cite{[FR]}.

\subsection{Proof of Proposition \ref{prop} for equation \eqref{6.1}}
\label{prop_rel}
We first consider the analog of Lemma \ref{l1.4.1}.
\begin{lemma}
\label{l1.4.1rel}
Let  $E>c^2$ and let $R_E$ and $C_E^{\rm rel}$ be defined by
\begin{equation}
C_E^{\rm rel}:=\min\Big({E-c^2\over 2\beta_0},{c^2\big(\big({E-c^2\over 4c^2}+1\big)^2-1\big)\over 4\beta_1 n \big({3(E-c^2)\over 2c^2}+1\big)^2}\Big),\
\sup_{|x|\ge R_E}(1+|x|)^{-\alpha}\le {C_E^{\rm rel}\over 2}.\label{111r}
\end{equation}
If $x(t)$ is a solution of equation (1.1) of energy $E$ such that $|x(0)|<R_E$ and if there exists a time $T>0$ such that $x(T)=R_E$ then
\begin{equation}
|x(t)|^2\ge R_E^2+  {1\over 2}{c^2\big(\big({E-c^2\over 4c^2}+1\big)^2-1\big)\over  \big({3(E-c^2)\over 2c^2}+1\big)^2}|t-T|^2\textrm{ for }t\in (T,+\infty),\label{114r}
\end{equation}
and there exists a unique $(x_+,v_+)\in \R^n\times\S_{E,c}$ so that
$x(t)=x_++tv_++y_+(t)$, $t\in \R$,
where $|y_+(t)|+|\dot y_+(t)|\to 0$ as $t\to +\infty.$
\end{lemma}
The proof of Lemma \ref{l1.4.1rel} is similar to the proof of Lemma \ref{l1.4.1}.

The solutions $x(t)$ of equation \eqref{6.1} in $B(0,R)$ for some $R>0$ also have properties \eqref{1.4a} and \eqref{1.5a}  at fixed and sufficiently large energy. Therefore at fixed and sufficiently large energy we consider the inverse kinematic problem in a ball $B(0,R)$ for equation \eqref{6.1} similar to the inverse kinematic problem given in Section \ref{invkin}. Then Lemma \ref{inv} still holds (see \cite[Theorem 1.2]{[Jo3]}) and the connection  between boundary data of the inverse kinematic problem and the scattering map $S_E^{\rm rel}$ is similar to the one given for the nonrelativistic case in Section \ref{connect} (note that the radius $R$ has also to be chosen so that $\sup_{|x|\ge R}(1+|x|)^{-\alpha}<c^2/(144\beta_1 n)=\lim_{E\to +\infty}C_E^{\rm rel} $). This proves Proposition \ref{prop} for equation \eqref{6.1}.\hfill $\Box$

\subsection{Proof of Theorem \ref{thm}  for equation \eqref{6.1}}
\label{thm_rel}
We assume that the electromagnetic field $(V,B)$ in equation \eqref{6.1} satisfies \eqref{1.3} and \eqref{1.4} and is so that $B\equiv 0$ and $V$ is spherically symmetric outside $B(0,R)$ for some $R>0$. Let $W\in C^2([R,+\infty),\R)$ be defined by $V(x)=W(|x|)$ for $x\not\in B(0,R)$. We give the analog of Lemma \ref{lemrmin}.
Let 
$\tilde \beta=(2\beta_0^2)^{1\over \alpha}\Big(E(2\beta_0+\beta_1)+\beta_1\beta_0-\Big((E(2\beta_0+\beta_1)+\beta_1\beta_0)^2-4\beta_0^2(E^2-c^4)\Big)^{1\over 2}\Big)^{-{1\over \alpha}}$ and set
\begin{equation}
\beta: =\max\left(\tilde \beta
,\left({\beta_0\over E-c^2}\right)^{1\over \alpha}E^{-1}c\sqrt{(E+\beta_0)^2-c^4}, {cR\sqrt{(E+\beta_0)^2-c^4}\over E}\right).\label{4k8} 
\end{equation}
Then for $q\ge\beta$ consider the real number $r_{\min,q}$ defined by 
\begin{equation}
r_{\min,q}=\sup\{r\in  (R,+\infty)\ | \ (E-W(r))^2-c^4-{q^2E^2\over c^2r^2}=0\}.\label{8.2}
\end{equation}
The function $r_{\min,.}$ has the following properties.

\begin{lemma}
\label{lemrmin2}
The function $r_{\min,.}$ is a $C^2$ strictly increasing function from $[\beta,+\infty)$ to $(R,+\infty)$ and we have  $r_{\min,q}^3W'(r_{\min,q}){c^2(E-W(r_{\min,q}))\over E^2}<q^2$ for $q\ge\beta$ and
\begin{eqnarray}
&&(E-W(r_{\min,q}))^2-c^4-{q^2E^2\over c^2r_{\min,q}^2}=0,\label{8.7a}\\
&&{d r_{\min,q}\over dq}={E^2qr_{\min,q}\over  -c^2(E-W(r_{\min,q}))r_{\min,q}^3W'(r_{\min,q})+q^2E^2}>0.\label{8.7b}
\end{eqnarray}
In addition the following estimates and asymptotics at $+\infty$ hold
\begin{equation}
 {qE\over c\sqrt{(E+\beta_0)^2-c^4}}\le r_{\min,q}\le {Eq\over c\sqrt{\Big(E-\beta_0 \big({qE\over c\sqrt{(E+\beta_0)^2-c^4}}\big)^{-\alpha}\Big)^2-c^4}},\label{8k6b}
\end{equation}
for $q\ge\beta$, and
\begin{equation}
r_{\min,q}={qE\over c\sqrt{E^2-c^4}}+O(q^{1-\alpha}),\textrm{ as }q\to +\infty.\label{8k6c}
\end{equation}
\end{lemma}
Then take any plane
 ${\cal P}$  containing $0$ and keep notations of Section \ref{scatangle}. 
Let $q\ge\beta$.
Then set $x_-:=(c\sqrt{1-c^4/E^2})^{-1}qe_1$ and $v_-=c\sqrt{1-c^4/E^2}e_2$, and consider $x_q(t)$ the solution of \eqref{6.1} with energy $E$ and with initial conditions \eqref{1.3} at $t\to -\infty$.  We write $x_q$ in polar coordinates: $x_q(t)=r_q(t)(\cos(\theta_q(t))e_1(t)+\sin(\theta_q(t))e_2)$ for $t\in (-\infty,t_-)$ where $t_-:=\sup\{t\in \R\ |\ |x_q(s)|\ge R \textrm{ for }s\in (-\infty,t)\}$ and the functions $r_q$, $\theta_q$, satisfy 
\begin{eqnarray}
&&\ddot r_q(t)={-W'(r_q(t))\over \left({E-W(r_q(t))\over c^2}\right)^3}+q^2E^2{E-W(r_q(t))-rW'(r_q(t))\over r^3(E-W(r_q(t)))^3},\label{9.3e}\\
&&{r_q(t)^2\dot \theta_q(t)\over \sqrt{1-{\dot r_q(t)^2+r_q(t)^2\dot \theta_q(t)^2\over c^2}}}={qE\over c^2}.\label{9.3f}
\end{eqnarray}
The energy $E$ defined by \eqref{6.2} is then written as follows
\begin{equation}
1-{\dot r_q(t)^2\over c^2}-{c^4\over (E-W(r_q(t)))^2}-{q^2E^2\over c^2r_q(t)^2(E-W(r_q(t)))^2}=0.\label{9.5}
\end{equation}
We also have 
\begin{equation}
\dot \theta_q(t)={qE\over r_q(t)^2(E-W(r_q(t)))}.\label{9.6}
\end{equation}
Similarly to Section \ref{scatangle} we have
 $t_-=+\infty$, $r_q(t_q+t)=r_q(t_q-t)$ for $t\in\R$, and $\pm\dot r_q(t)>0$ for
$\pm t>t_q$ where $t_q=\inf\{t\in (-\infty,t_-)\ |\ r_q(t)=r_{\min,q}\}$.

We thus define
\begin{equation}
g(q)=\int_{-\infty}^{+\infty}{dt\over r_q(t)^2(E-W(r_q(t)))}=2\int_{t_q}^{+\infty}{dt\over r_q(t)^2(E-W(r_q(t)))}.\label{9.0}
\end{equation}
From \eqref{9.6} $Eqg(q)=\int_\R\dot \theta_q(s) ds$ is the scattering angle of $x_q(t)$, $t\in \R$, and
\begin{equation}
S_{1,E}^{\rm rel}(v_-,x_-)=c\sqrt{1-c^4/E^2}\Big(\cos\big(Eqg(q)+{\pi\over 2}\big),\sin\big(Eqg(q)+{\pi\over 2}\big)\Big),\label{9.104a}
\end{equation}
for $q\in [\beta,+\infty)$ and where $S_{1,E}^{\rm rel}$ is the first component of the scattering map $S_E^{\rm rel}$.
Since $qg(q)\to 0$ as $q\to +\infty$ and  that $g$ is continuous on $[\beta,+\infty)$  $S_{1,E}^{\rm rel}$ uniquely determines the function $g$. 

We now provide the reconstruction formulas for $W$ from $g$ in a neighborhood of infinity .
Let $\chi$ be the strictly increasing function from $[0,\beta^{-2})$ to $[0,r_{\min,\beta}^{-1})$, continuous on $[0,\beta^{-2})$ and $C^2$ on $(0,\beta^{-2})$, defined by
$\chi(0)=0$ and $\chi(\sigma)=r_{\min,\sigma^{-{1\over 2}}}^{-1}$ for $\sigma\in (0,\beta^{-2})$.
Let $\phi:(0,\chi(\beta^{-2}))\to (0,\beta^{-2})$ denote the inverse function of $\chi$.
From \eqref{8.7a} it follows that
$(E-V(\chi(u)^{-1}))^2-c^4-{E^2\chi(u)^2\over c^2 u}=0$ for $u\in (0,\beta^{-2})$, and
$\big((E-V(s^{-1}))^2-c^4\big)\phi(s)-{E^2s^2\over c^2 }=0$ for $s\in (0,\chi(\beta^{-2}))$.

Define the function $H$ from $(0,\beta^{-2})$ to $\R$ by \eqref{4.11a}.
The following formulas are valid
\begin{equation}
H(\sigma)
=\pi\int_0^{\chi(\sigma)}{ds\over \sqrt{(E-V(s^{-1})^2-c^4}},\ 
{E\over c\pi\sqrt{\sigma}}{dH\over d\sigma}(\sigma)={d\over d\sigma}\ln(\chi(\sigma)),\label{4k11c}
\end{equation}
for $\sigma\in (0,\beta^{-2})$. The proof of formulas \eqref{4k11c} is similar to the proof of formulas \eqref{4.11c}.

Then note that from \eqref{8k6c} it follows that
$\chi(\sigma)=\Big({E\over c\sqrt{E^2-c^4}}\sigma^{-{1\over 2}}+O(\sigma^{-1+\alpha\over 2})\Big)^{-1}$ $={c\sqrt{E^2-c^4}\over E}\sigma^{1\over 2}+O(\sigma^{\alpha+1\over 2}),$ $\sigma\to 0^+$, and 
\begin{equation}
\ln \Big({E\chi(\sigma)\over c\sqrt{E^2-c^4}\sigma^{1\over 2}}\Big)=\ln(1+O(\sigma^{\alpha\over 2}))\to 0,\ \textrm{as}\ \sigma\to 0^+.\label{4k10b}
\end{equation}

Therefore we obtain the following reconstruction formulas
\begin{eqnarray}
\chi(\sigma)&=&{c\sqrt{E^2-c^4}\over E}\sigma^{1\over 2}e^{\int_0^\sigma\big({E\over c\pi\sqrt{s}}{dH\over ds}(s)-{1\over 2s}\big)ds},\ \textrm{for}\ \sigma\in (0,\beta^{-2}),
\label{4k12a}\\
W(s)&=&E- \left(c^4+{E^2\over c^2s^2\phi(s^{-1})}\right)^{1\over 2},\ \textrm{for}\ s\in(r_{\min, \beta},+\infty).\label{4k12b}
\end{eqnarray}
Set
$\beta'={\beta\over \Big( 2E-2\beta_0\beta^{-\alpha}(2\beta_0+2E)^{\alpha\over2}\Big)^{1\over 2}}$.
Then note that from \eqref{4k8} and \eqref{8k6b} it follows that
$r_{\min,\beta}\le \beta'$.
Therefore using \eqref{4k12a} and \eqref{4k12b} we obtain that $W$ is determined by the first component of the scattering map $S_E^{\rm rel}$ on $(\beta',+\infty)$.

The proof of Theorem \ref{thm} for equation  \eqref{6.1} then relies on this latter statement and on Proposition \ref{prop} for equation \eqref{6.1}. \hfill $\Box$

\appendix

\section{Proof of Lemmas \ref{l1.4.1} and \ref{lemrmin}}
In this Section we give a proof of Lemmas \ref{l1.4.1} and \ref{lemrmin}, and we give details on the derivation of formulas \eqref{4.11c} and the equality "$q=q'$" in Section \ref{scatangle}.

\begin{proof}[Proof of Lemma \ref{l1.4.1}]
We will use the following estimate.
Under conditions \eqref{1.3} and \eqref{1.4} we have
\begin{eqnarray}
\label{ch1.2.17a}|F(x,v)|&\le&\beta_1n (1+|x|)^{-(\alpha+1)}(1+|v|),\textrm{ for }(x,v)\in\R^n\times\R^n.
\end{eqnarray}
Let
\begin{equation}
I(t) = {1\over 2} |x(t)|^2.\label{120}
\end{equation}
Then using the conservation of energy  and equation (1.1) we have
\begin{eqnarray}
\label{ch1.4.5a}\dot I(t) &=&x(t)\cdot \dot x(t),\\
\label{ch1.4.6b}\ddot I(t)&=&2E -2V(x(t))- x(t)\cdot F(x(t),\dot x(t)),
\end{eqnarray}
for $t\in \R$.
Using \eqref{ch1.2.17a} and estimate on $V$ and $|\dot x|(t)=\sqrt{2(E-V(x(t)))}\le \sqrt{2(E+\beta_0)}$ for $t\in \R$, we obtain that
\begin{eqnarray}
\ddot I(t)&\ge& 2E -2\beta_0(1+|x(t)|)^{-\alpha}- n\beta_1(1+\sqrt{2(E+\beta_0)})|x(t)|(1+|x(t)|)^{-\alpha-1}\nonumber\\
&\ge&(2\beta_0+n\beta_1) (1+\sqrt{2(E+\beta_0)}) (C_E- (1+|x(t)|)^{-\alpha}),\label{121}
\end{eqnarray}
for $t\in \R$.
Hence  we have
\begin{equation}
\ddot I(t)>(2\beta_0+n\beta_1) (1+\sqrt{2(E+\beta_0)}) {C_E\over 2}=E\textrm{ whenever }|x(t)|\ge R_E.\label{122}
\end{equation}
Let $t_0=\inf\{t\in [0,T]\ |\ |x(t)|>R_E\}$. Then using \eqref{120} we have
\begin{equation}
\dot I(t_0)=\lim_{h\to 0^+}{I(t_0)-I(t_0-h)\over h}\ge 0.\label{123}
\end{equation}
Combining \eqref{122} and \eqref{123} we obtain that $|x(t)|\ge R_E$, $\dot I(t)\ge 0$, $\ddot I(t)\ge E$ for $t\in [t_0,+\infty)$.
Moreover we have
\begin{equation}
I(t)=I(t_0)+\dot I(t_0)(t-t_0)+\int_{t_0}^t(s-t_0)\ddot I(s)ds
\ge {1\over 2}R_E^2+{E\over 2}(t-t_0)^2,\label{124}
\end{equation}
for $t\in [t_0,+\infty)$. This proves \eqref{114}.
Then using that 
$|\dot x(t)|\le\sqrt{2(E+\beta_0)}$ for $t\in \R$  and using \eqref{114} we have
\begin{equation}
\label{ch1.4.22}|F(x(\tau),\dot x(\tau))|\le n\beta_1(1+\sqrt{E}|\tau-T|)^{-(\alpha+1)}(1+\sqrt{2(E+\beta_0)}),
\end{equation}
for $\tau\in [T,+\infty [$. Equation (1.1) then gives
\begin{equation}
x(t)=x_++tv_++y_+(t),\label{ch1.4.23}
\end{equation}
for $t\in (0,+\infty)$, where
\begin{eqnarray}
\label{ch1.4.24a}v_+&=&\dot x(0)+\int_0^{+\infty}F(x(\tau),\dot x(\tau))d\tau,\\
\label{ch1.4.24b}x_+&=&x(0)-\int_0^{+\infty}\int_\sigma^{+\infty} F(x(\tau),\dot x(\tau))d\tau
d\sigma\\
\label{ch1.4.24c}y_+(t)&=&\int_t^{+\infty}\int_\sigma^{+\infty}F(x(\tau),\dot x(\tau)))d\tau d\sigma,
\end{eqnarray}
for $t\in (0,+\infty)$, where
by \eqref{ch1.4.22}  the integrals in \eqref{ch1.4.24a}, \eqref{ch1.4.24b} and \eqref{ch1.4.24c} are absolutely convergent ($\alpha >1$) and  $|y_+(t)|+|\dot y_+(t)|\to 0$ as $t\to +\infty$.
\end{proof}

\begin{proof}[Proof Lemma \ref{lemrmin}]
Note that for $q\ge\beta$ we have 
${q^2/(2R^2)}> E+\beta_0(1+R)^{-\alpha}\ge E-W(R)$ and $\lim_{r\to +\infty}q^2/(2r^2)=0<E=\lim_{r\to +\infty}(E-W(r))$. 
Hence using \eqref{4.2} we obtain $r_{\min,q}\in (R,+\infty)$ and
\begin{equation} 
W(r_{\min,q})+{q^2\over 2r_{\min,q}^2}=E,\label{4.2b}
\end{equation}
for $q\in [\beta,+\infty)$.

Let $q\in [\beta,+\infty)$.
From \eqref{4.2b}  it follows that
$q^2/(2r_{\min,q}^2)\le E+\beta_0$ and then $r_{\min,q}\ge q/\sqrt{2(E+\beta_0)}$.
Combining this latter estimate and estimate \eqref{1.3W} we obtain that
\begin{equation}
2\Big(E-\sup_{r\in \big({q\over \sqrt{2(E+\beta_0)}},+\infty\big)}W(r)\Big)\le{q^2\over r_{\min,q}^2}\le 2\Big(E-\inf_{r\in \big({q\over \sqrt{2(E+\beta_0)}},+\infty\big)}W(r)\Big).\label{4.100}
\end{equation}
Then using again \eqref{1.3W} we obtain
\begin{equation}
\sup_{r\in \big({q\over \sqrt{2(E+\beta_0)}},+\infty\big)}|W(r)| \le \beta_0 \left({q\over \sqrt{2(E+\beta_0)}}\right)^{-\alpha}.\label{4.101}
\end{equation}
Combining \eqref{4.100} and \eqref{4.101} we obtain
\begin{equation}
{q\over \sqrt{2E+2\beta_0 {(2(E+\beta_0))^{\alpha\over 2}\over q^\alpha}}}\le r_{\min,q}\le {q\over \sqrt{2E-2\beta_0 {(2(E+\beta_0))^{\alpha\over 2}\over q^\alpha}}}.\label{4.102}
\end{equation}
Estimates \eqref{4.6a} and the asymptotics \eqref{4.6c} follows from \eqref{4.102}.

Note that using \eqref{1.3W} we have
\begin{eqnarray}
rW'(r)&< & \beta_1r^{-\alpha}\le \beta_1 q^{-\alpha}(2(E+\beta_0))^{\alpha\over 2},\label{4.103a}\\
2E-2W(r)&> & 2E-2\beta_0 r^{-\alpha}\ge 2E- 2\beta_0 q^{-\alpha}(2(E+\beta_0))^{\alpha\over 2},\label{4.103b}
\end{eqnarray}
for $r\ge q(2(E+\beta_0))^{-{1\over 2}}$ and $q\ge\beta$. From \eqref{4.103a} and \eqref{4.103b} we obtain 
\begin{equation}
rW'(r)<2E-2W(r),\label{4.103c}
\end{equation}
for $r\ge q(2(E+\beta_0))^{-{1\over 2}}$ and $q\ge\beta$.
Then consider the function $f\in C^2([\beta,+\infty)\times (R,+\infty),\R)$ defined by $f(q,r)=2E-2W(r)-{q^2\over 2r^2}$ for $(q,r)\in [\beta,+\infty)\times (R,+\infty)$ (${\pa f\over \pa r}(q,r)=-2W'(r)+2q^2/r^3$). We have $f(q,r_{\min,q})=0$ for $q\ge\beta$ and from the implicit function theorem and \eqref{4.102} it follows that $r_{\min,.}$ is a $C^2$ strictly increasing function from $[\beta,+\infty)$ to $(R,+\infty)$ so that $r_{\min,q}^3W'(r_{\min,q})<q^2$ for $q\ge\beta$ and the derivative of $r_{\min,.}$ is given by \eqref{4.7b}.
\end{proof}

\begin{proof}[Derivation of formulas \eqref{4.11c}]
We first make the change of variables $"r"=r_q(t)$ in \eqref{4.0} ($dr = \dot r_q(t) dt$ and $\dot r_q(t)=\sqrt{{2E-{q^2\over r_q(t)^2}-2V(r_q(t))}}$, see \eqref{2.5}) and we obtain 
\begin{equation}
g(q)=2\int_{r_{\min,q}}^{+\infty} { dr\over r^2\sqrt{2E-{q^2\over r^2}-2V(r)}},\ \textrm{ for }\ q>\beta.\label{4.9}
\end{equation}
Performing the change of variables ``$r^{-1}=s$'' in \eqref{4.9} we obtain
\begin{equation}
g(u^{-{1\over 2}})=2\int_0^{\chi(u)} { ds\over \sqrt{2E-{s^2\over u}-2V(s^{-1})}},\ \textrm{ for }\ u\in (0,\beta^{-2}).\label{4.200}
\end{equation}
Let $\sigma\in (0,\beta^{-2})$. From \eqref{4.200} and \eqref{4.11a} it follows that
\begin{equation}
H(\sigma)=\int_0^{\chi(\sigma)} \left(\int_{\phi(s)}^\sigma{du\over \sqrt{\sigma-u}\sqrt{2(E-V(s^{-1}))u-s^2}}\right)ds.\label{4.202}
\end{equation}
And performing the change of variables $u=\phi(s)+\ep(\sigma-\phi(s))$ in \eqref{4.202} ($du=(\sigma-\phi(s))d\ep$)
and using the equality \eqref{4.203}  we obtain
\begin{eqnarray}
&&\int_{\phi(s)}^\sigma{du\over \sqrt{\sigma-u}\sqrt{2(E-V(s^{-1}))u-s^2}}
=\phi(s)^{1\over 2}s^{-1}\int_{\phi(s)}^\sigma{du\over \sqrt{\sigma-u}\sqrt{u-\phi(s)}}\nonumber\\
&&=\phi(s)^{1\over 2}s^{-1}\int_0^1{d\ep\over \sqrt{\ep}\sqrt{1-\ep}}=\phi(s)^{1\over 2}s^{-1}\pi,\label{4.204}
\end{eqnarray}
for $s\in (0,\chi(\beta^{-2}))$ (we used the integral value $\pi=\int_0^1{d\ep\over \sqrt{\ep}\sqrt{1-\ep}}$).
Using \eqref{4.202}, \eqref{4.204} and \eqref{4.203} we obtain the first equality in \eqref{4.11c}
\begin{equation}
H(\sigma)=\pi \int_0^{\chi(\sigma)} \phi(s)^{1\over 2}s^{-1}ds=\pi \int_0^{\chi(\sigma)} (2(E-V(s^{-1})))^{-{1\over 2}}ds.\label{4.205}
\end{equation}
From \eqref{4.205} it follows that
\begin{equation}
{dH\over d\sigma}(\sigma)={\pi\over \sqrt{2(E-V(\chi(\sigma)^{-1}))}}{d\chi\over d\sigma}(\sigma).\label{4.206}
\end{equation}
Then combining \eqref{4.203b} and \eqref{4.206} we obtain the second equality in \eqref{4.11c}.
\end{proof}

We end this appendix by giving details on the equality $q'=q$ in section \ref{scatangle}. We keep the notations of section \ref{scatangle}.
We set $u_\theta= \big({x\over r}\big)^\bot$ and we have
\begin{equation}
r^2\dot\theta u_\theta=r \dot x-\dot r x=r (v_-+\dot y_-)-\dot r (x_-+tv_-+y_-).\label{2.6a}
\end{equation}
Using $x_-\cdot v_-=0$ and 
$y_-=o(1)$, $\dot y_-=o(t^{-1})$, as $t\to -\infty$, we have
\begin{eqnarray}
r(t)&=&|x_-+tv_-+y_-|=\big(t^2|v_-|^2+2t y_-\cdot v_-+|x_-+y_-|^2\big)^{1\over 2}\nonumber\\
&=&-t|v_-|+o(1),\ t\to -\infty,\label{2.6c}\\
\dot r(t)&=&{(v_-+\dot y_-)\cdot (x_-+tv_-+y_-)\over r(t)}={t|v_-|^2+o(1)\over -t|v_-|+o(1)}\nonumber\\
&=&- |v_-|+o(t^{-1}),\ t\to -\infty,\label{2.6d}
\end{eqnarray}
and we obtain
\begin{eqnarray}
r^2\dot\theta u_\theta&=&(-|v_-|t+o(1))(v_-+\dot y_-)-(-|v_-|+o(t^{-1}))(x_-+tv_-+y_-)\nonumber\\
&=& |v_-|x_-+o(1),\ t\to -\infty,\label{2.6e}
\end{eqnarray}
\begin{equation}
u_\theta=\Big({x(t)\over r(t)}\Big)^\bot=\Big({x_-+tv_-+o(1)\over - t|v_-|+o(1)}\Big)^\bot=-\widehat{v_-}^\bot+o(1),\ t\to -\infty\label{2.6f}
\end{equation}
where  $\hat w={w\over |w|}$ for $w\not=0$. Using \eqref{2.6e} and \eqref{2.6f} we obtain
\begin{equation}
r^2\dot\theta u_\theta\cdot \widehat{v_-}^\bot =x_-\cdot v_-^\bot+o(1),\ 
 u_\theta\cdot \widehat{v_-}^\bot=-1+o(1),\ t\to -\infty,\label{2.6i}
\end{equation}
which proves that 
$
q'=r^2\dot \theta=-x_-\cdot v_-^\bot.
$\hfill  $\Box$

\vskip 8mm

\noindent A. Jollivet

\noindent Laboratoire de Physique Th\'eorique et Mod\'elisation,

\noindent CNRS UMR 8089/Universit\'e de Cergy-Pontoise

\noindent 95302 Cergy-Pontoise, France

\noindent e-mail: alexandre.jollivet@u-cergy.fr

\end{document}